\documentclass[a4paper,USenglish,cleveref, autoref, thm-restate,]{lipics-v2021}

\PassOptionsToPackage{nameinlink}{cleveref}

\hideLIPIcs
\nolinenumbers

\usepackage{braket}
\usepackage{thmtools}
\usepackage{booktabs}

\title{On the Limits of Quantum Multiparty Simultaneous Communication}
\titlerunning{On the Limits of Quantum Multiparty Simultaneous Communication}

\author{Pedro Montealegre}{Universidad Adolfo Ibáñez, Santiago, Chile}{p.montealegre@uai.cl}{}{}
\author{Iván Rapaport}{Universidad de Chile, Santiago, Chile}{rapaport@dim.uchile.cl}{}{}
\author{Jorge Valenzuela}{Universidad Adolfo Ibáñez, Santiago, Chile}{jorg.valenzuela@alumnos.uai.cl}{}{}

\authorrunning{P. Montealegre, I. Rapaport, J. Valenzuela}

\Copyright{Pedro Montealegre, Iván Rapaport and Jorge Valenzuela}



\relatedversion{}

\newtheorem{question}{Question}
\newtheorem{open}{Open problem}

\begin{document}
\date{2026}
\maketitle

\begin{abstract}
We study the relative power of public randomness and quantum
communication in the multiparty Simultaneous Message Passing
(SMP) model. By analyzing a total Boolean function and a relational problem,
we establish polynomial and exponential separations between
classical public-coin protocols and quantum protocols without
shared randomness or prior entanglement.

We first introduce XOR Equality, a total Boolean function in
which each of $k$ players receives an unrestricted string
$x_j\in\{0,1\}^n$, and the referee must decide whether
$x_1\oplus\cdots\oplus x_k=0^n$. For $k\ge8$ and worst-case
error at most $1/3$, we prove quantum lower bounds of
$\Omega(\min\{n,k^2\})$ on total communication and
$\Omega(\min\{n/k,k\})$ on maximum message length.
A classical public-coin protocol uses two bits per player
with error at most $1/4$. For $n=k^2$ and worst-case error at most $1/3$, we prove
that quantum protocols require maximum message length
$\Omega(\sqrt n)$ qubits, whereas two public-coin bits
per player suffice.

We then turn to relations, for which we obtain exponential
separations. Specifically, we study Multiparty Index Coordination,
a $k$-player generalization of the two-player relation studied by
Gavinsky, Kempe, Regev, and de Wolf (STOC 2006).
We prove
an unambiguous quantum lower bound of $\Omega(n^{1-1/k})$
on maximum message length, matched by a classical private-coin
protocol. For any fixed error $\varepsilon\in(0,1/8)$,
we also prove a bounded-error quantum lower bound of
$\Omega(n^{(k-1)/(k+1)})$ while we show that public-coin protocols solve the
relation unambiguously with maximum message length
$O(\log n)$, giving exponential separations in both regimes.
The implicit constants are independent of $k$, and both
quantum lower bounds become linear when $k\ge c\log n$
for any fixed $c>0$.
\end{abstract}

\clearpage
\setcounter{page}{1}

\section{Introduction}

Shared randomness allows distributed players to coordinate their messages without communicating with one another. Quantum communication offers a different resource: each player can encode an input into a quantum state, and a referee can perform a joint measurement on the received messages. Understanding when one resource can substitute for the other is a central question in quantum communication complexity. In this paper, we study this question in the multiparty setting, through a total Boolean function and a relational coordination problem.

In the Simultaneous Message Passing (SMP) model, introduced by Yao~\cite{yao1979}, each of $k$ players receives a private input and sends a single message to a referee. The players cannot communicate with one another, and the referee must produce an output from their messages alone. In this model, we compare classical protocols with public randomness against quantum protocols without shared randomness or prior entanglement. In the public-coin model, the randomness is shared among the players; the referee has no access to it. In the quantum model, the referee may perform an arbitrary joint measurement over the registers sent by the players. We measure the communication as the maximum message length, denoted by $m=\max_j m_j$, where $m_j$ is the worst-case length of player $P_j$'s message.

The two-player \textsc{Equality} function illustrates the power
of both public randomness and quantum messages. It admits a
classical public-coin protocol with constant communication,
whereas classical private-coin protocols require
$\Omega(\sqrt n)$ communication~\cite{ambainis1996,Newman96}.
Quantum fingerprinting reduces the communication to
$O(\log n)$ qubits~\cite{buhrman2001}. In the multiparty setting,
Le Gall, Nadler, Nishimura, and Oshman~\cite{legall2024modified}
show that quantum messages can replace public randomness
for protocols based on modified equality queries.
These results raise a broader question: can quantum messages
always reproduce the communication savings provided by
public randomness?

A multiparty separation is also known for promise functions. Gavinsky, Ito, and Wang~\cite{gavinsky2013} introduced the Gap-Parity problem $GP_k$, for which one public-coin bit per player suffices, while quantum protocols require total communication $\Omega(k n^{1-2/k})$. Thus, the existence of a separation beyond two players is already established. We investigate two complementary aspects of this phenomenon: whether a separation persists for a total Boolean function, and how much stronger it can become for a relational coordination task as more players participate.

Our first question concerns total Boolean functions. A total Boolean function assigns a unique output bit to every input tuple: there is neither an input promise nor a choice among several valid outputs. Establishing a multiparty separation under this restriction would show that the advantage of public randomness does not depend on either feature.

\begin{question}\label{q:total-boolean}
    Can public randomness yield an advantage in maximum message length over quantum protocols for a total Boolean function in the multiparty setting?
\end{question}

For relational problems, exponential separations are already known. Gavinsky, Kempe, and de Wolf~\cite{gavinsky2004} introduced a two-party problem (we will call it \textsc{Index Coordination} from now on), which admits an $O(\log n)$ public-coin protocol but requires $\Omega(\sqrt n)$ quantum communication in the unambiguous regime. In this regime, the referee may abort (outputting a special symbol) with bounded probability but must never return an incorrect answer. Gavinsky, Kempe, Regev, and de Wolf~\cite{gavinsky2006} subsequently proved an $\Omega(n^{1/3})$ lower bound in the bounded-error regime, where the referee may return an incorrect answer with bounded probability. These results show that even a joint quantum measurement cannot always compensate for the absence of common random choices.

Turning to relations, we ask whether relaxing the requirement of a unique Boolean output allows a stronger separation that reflects the number of players. Adding players who do not affect the output would trivially preserve a two-player lower bound without revealing anything about multiparty coordination. We instead seek an extension of \textsc{Index Coordination} in which every input contributes to the required output and the lower bounds reflect the number of participating players.

\begin{question}
    Does \textsc{Index Coordination} admit a natural multiparty generalization that involves every player nontrivially and preserves the exponential separation?
\end{question}

\begin{question}
    For such a generalization, can the quantum lower bound become stronger, approaching a linear bound, in both the unambiguous and bounded-error regimes?
\end{question}

\subsection{Our Contributions}

We answer all three questions affirmatively. Our Boolean result establishes a polynomial separation for \textsc{XOR Equality} with a growing number of players. By turning to relations, we obtain exponential separations: our relational result extends \textsc{Index Coordination} to $k$ players and gives quantum lower bounds that strengthen with $k$. Both arguments study the information available in conditional average messages, using quantum information bounds and the behavior of parities across independent quantum encodings.\\

\noindent{\bf Polynomial separation for a total Boolean function.}
We study the Boolean function $\operatorname{XEQ}_{k,n}$, in which the players receive unrestricted strings $x_1,\ldots,x_k\in\{0,1\}^n$ and the referee decides whether
\[
    x_1\oplus\cdots\oplus x_k=0^n.
\]

We show a public-coin protocol using two bits per player with error at most $1/4$. In contrast, for $k\ge8$, every quantum protocol without shared randomness or prior entanglement and with error at most $1/3$ must send
\[
    \Omega(\min\{n/k,k\}) \text{ qubits of maximum communication.}
\]
Taking $n=k^2$ yields $O(1)$ public-coin maximum message length versus $\Omega(\sqrt n)$ qubits. This polynomial separation uses a growing number of players.\\

\noindent{\bf Exponential separation for a relational problem.}
We then relax the class of problems considered and introduce the relation $\operatorname{IC}_{k,n}$, where $k\ge2$ and $n$ is even. Each player $P_j$ receives a string $x_j\in\{0,1\}^n$, and the last player additionally receives a selector string $s\in\{0,1\}^n$ of Hamming weight $n/2$. The referee must output
\[
    (i,x_1^{(i)},\ldots,x_k^{(i)})
    \qquad\text{for some }i\text{ such that }s^{(i)}=1.
\]
This recovers the two-player relation from \cite{gavinsky2006} when $k=2$. For arbitrary $k$, every player contributes a bit to the output, while only the last player knows which coordinates are valid.

Public randomness lets the players sample the same candidate coordinates. A constant number of samples gives an unambiguous protocol with maximum message length
\[O(\log n) \text{ bits }\]
for any fixed positive abort probability. Without shared randomness or prior entanglement, we prove quantum lower bounds of
\[
    \Omega\!\left(n^{(k-1)/(k+1)}\right)
    \quad\text{qubits in the bounded-error regime, and}
\]
\[
    \Omega\!\left(n^{1-1/k}\right)
    \quad\text{qubits in the unambiguous regime}.
\]
A classical private-coin protocol matches the unambiguous bound. Thus, for this relation, quantum messages offer no asymptotic improvement over private coins in the unambiguous setting.

For fixed error parameters, the exponential separation holds for every integer-valued function $k=k(n)\ge2$. Moreover, when $k\ge c\log n$ for any fixed $c>0$, both quantum lower bounds become $\Omega(n)$ and match the protocol that sends the complete inputs and selector.

\begin{table}[htbp]
\centering
\renewcommand{\arraystretch}{1.3}
\resizebox{\textwidth}{!}{
\begin{tabular}{@{}lllll@{}}
\toprule
\textbf{Type} & \textbf{Problem} & \textbf{Model} & \textbf{Error Regime} & \textbf{Max Message Length} \\
\midrule
Total Boolean & $\operatorname{XEQ}_{k,n}$ & Public-Coin & Bounded-Error & $O(1)$ bits \\
Total Boolean & $\operatorname{XEQ}_{k,n}$ & Quantum& Bounded-Error & $\Omega(\min\{n/k, k\})$ qubits, $k\geq 8$ \\
\midrule
Relational & $\operatorname{IC}_{k,n}$ & Public-Coin & Unambiguous \& Bounded & $O(\log n)$ bits \\
Relational & $\operatorname{IC}_{k,n}$ & Private-Coin & Unambiguous \& Bounded & $O(n^{1-1/k})$ bits \\
Relational & $\operatorname{IC}_{k,n}$ & Quantum& Unambiguous & $\Theta(n^{1-1/k})$ qubits \\
Relational & $\operatorname{IC}_{k,n}$ & Quantum& Bounded-Error & $\Omega\left( n^{\frac{k-1}{k+1}} \right)$ qubits\\
\bottomrule
\end{tabular}
}
\vspace{0.2cm}
\caption{Summary of asymptotic bounds on maximum message length in the multiparty SMP model. The polynomial separation for $\operatorname{XEQ}_{k,n}$ uses a growing number of players, while the exponential separation for $\operatorname{IC}_{k,n}$ holds uniformly for any $k \ge 2$.}
\label{tab:results}
\end{table}

\noindent{\bf Proof ideas.}
The Boolean and relational lower bounds share an information-theoretic common argument: a short quantum message cannot reliably reveal too many independent bits of its input. The proofs use this constraint in three ways.

\begin{enumerate}
    \item \textbf{Total Boolean Function.}
    For $\operatorname{XEQ}_{k,n}$, we apply the conditional-message
    approach to arbitrary linear forms of each player's input.
    A random-access-code argument bounds the dimension of the span
    of parities that each message predicts with sufficiently large
    advantage. When total communication is small, there is a nonzero
    vector orthogonal to all these spans. We compare the average
    messages conditioned on global XOR zero and on this nonzero
    vector. Parity factorization, together with an aggregate parity
    estimate from Ben-Aroya, Regev, and de Wolf~\cite{benaroya2008hypercontractive},
    shows that the two ensembles are too close for the referee to
    distinguish with the required accuracy.

    \item \textbf{Relational Bounded-Error.}
    For $\operatorname{IC}_{k,n}$, random-access-code bounds identify
    a set of coordinates on which the first $k-1$ players collectively
    reveal little information. We choose the selector to mark these
    coordinates. Recovering all $k$ input bits at a selected coordinate
    necessarily recovers the parity of the first $k-1$ bits together
    with the last bit. The distinguishability of the two parity states
    factors into a product of local distinguishabilities. This reduces
    the analysis to the two-register asymmetric identification bound
    of Gavinsky et al.~\cite{gavinsky2006}, yielding the exponent
    $(k-1)/(k+1)$.

    \item \textbf{Relational Unambiguous.} In the unambiguous regime, for $\operatorname{IC}_{k,n}$, we bound the probability of identifying a prescribed tuple of labels encoded in a product of mixed states. This theorem applies to an arbitrary number of registers and arbitrary finite alphabets, and bounds this probability by the product of the corresponding optimal local identification probabilities. Combining this bound with random-access-code inequalities and an information-pooling argument gives the tight exponent $1-1/k$ for $\operatorname{IC}_{k,n}$.
\end{enumerate}

The paper proceeds as follows. We discuss related work in
\Cref{sec:related}, then introduce the communication models and
state-identification tools. We establish the total Boolean separation
in \Cref{sec:xeq}, then turn to relations to obtain the exponential
separations in \Cref{sec:separation}, before concluding with open problems.

\section{Related Work}
\label{sec:related}

Our work resides at the intersection of quantum communication complexity and quantum state discrimination~\cite{barnett2009} and builds upon several major historical milestones in the SMP model and other classical-versus-quantum separations in various distributed settings. 

\noindent{\bf Two-Party Foundational Separations.}
Bar-Yossef, Jayram, and Kerenidis~\cite{bar-yossef2004} showed that quantum communication can exponentially outperform classical public coins in the two-player SMP model using a restricted version of \textsc{Hidden Matching}. The separation in the opposite direction was established by Gavinsky et al.~\cite{gavinsky2004, gavinsky2006} using the two-player \textsc{Index Coordination} problem. Together, these results show that the two models are incomparable for relational problems. Our results give a natural multiparty generalization of the latter separation and extend it to total boolean functions.

\noindent{\bf Multiparty Shared Randomness vs. Quantum Communication.}
In the multiparty domain, Gavinsky, Ito, and Wang~\cite{gavinsky2013} previously established that shared randomness can exponentially outperform unentangled quantum messages. However, our work differs conceptually and methodologically from theirs. They studied a promise function, $GP_k$ (Gap-Parity), and used a hybrid argument over trace distances to prove a total-communication lower bound of $\Omega(k n^{1-2/k})$. We first consider the total Boolean function $\operatorname{XEQ}_{k,n}$ in \Cref{sec:xeq}. This removes the input promise while giving a polynomial separation in maximum message length for growing $k$. We then study a natural extension of the relational coordination problem ($\operatorname{IC}_{k,n}$), obtaining exponential separations. Moreover, our unambiguous bound is tight given a private-coin protocol that trivially yields a quantum protocol.

\noindent{\bf Multipartite Entanglement vs. Unentangled Messages.}
Very recently, Chakraborty, Banik, and de Wolf~\cite{chakraborty2026} demonstrated that SMP protocols utilizing shared \emph{multipartite entanglement} (such as GHZ states) can exponentially outperform unentangled quantum communication. Their quantum protocols require active, prior quantum entanglement distributed among the players. Our work, by contrast, studies the limits of unentangled quantum messages compared with classical shared randomness.

\noindent{\bf Separations in the Number-on-Forehead Model.}
In the number-on-forehead (NOF) model, each player sees all input blocks except its own. Yang and Zhang~\cite{yang2025simultaneousnof} gave a three-player simultaneous NOF relation with $O(\log n)$ quantum communication versus $\Omega(n^{1/16})$ randomized classical communication. Their subsequent one-way NOF separation~\cite{yang2026onewaynof} achieves $O(\log n)$ quantum cost versus $\Omega(n^{1/3}/2^{k/3})$ randomized cost.

\noindent{\bf Separations for Locally Checkable Problems.}
In a very recent breakthrough, Flin, Li, and Suomela~\cite{flin2026symmetry} gave an $O(1)$-round quantum algorithm for $3$-coloring cycles with high probability, compared with $\Theta(\log^* n)$ classical rounds. Their result extends to all LCLs with classical complexity $O(\log^* n)$ on bounded-degree graphs, including maximal independent set and maximal matching, giving the first natural graph problems with an asymptotic quantum advantage in the LOCAL model. More closely related to our comparison, Balliu et al.~\cite{balliu2025sharedrandomness} gave an LCL solvable in $O(\log n)$ classical rounds with shared randomness but requiring $\Omega(\sqrt n)$ rounds even in quantum-LOCAL without a preshared state. Thus shared randomness can also outperform unassisted quantum communication in locally checkable distributed tasks.

\section{Preliminaries}
\label{sec:preliminaries}

This section fixes the notation, communication models, and information-theoretic tools used below.

We identify $\{0,1\}^n$ with $\mathbb F_2^n$, where addition is coordinatewise XOR, and write $\langle a,x\rangle=\bigoplus_{i=1}^n a_i x_i$ for the inner product over $\mathbb F_2$. The vector $e_i$ denotes the $i$-th standard basis vector and $[n]$ is defined as $\{1,\dots,n\}$.

\subsection{Quantum Information Basics}
We restrict our attention to finite-dimensional Hilbert spaces, denoted by $\mathcal{H}$. The set of linear operators on $\mathcal{H}$ is denoted by $\mathrm{L}(\mathcal{H})$, and the set of density operators (positive semidefinite operators with unit trace) is denoted by $\mathrm{D}(\mathcal{H})$. A quantum state is represented by a density operator $\rho \in \mathrm{D}(\mathcal{H})$. If $\rho$ is a rank-one projector, it represents a pure state.

For any operator $A \in \mathrm{L}(\mathcal{H})$, the support $\mathrm{supp}(A)$ is the orthogonal complement of its kernel, $\ker(A)$. If $S \subseteq \mathcal{H}$ is a closed subspace, $S^\perp$ denotes its orthogonal complement, and $\Pi_S$ is the orthogonal projector onto $S$.

To quantify the distinguishability of two quantum states, we rely on the trace norm. For an operator $A$, the trace norm is defined as $\|A\|_{\mathrm{tr}} = \frac{1}{2} \operatorname{Tr}(\sqrt{A^\dagger A})$. Under this convention, the trace distance between two quantum states $\rho_0$ and $\rho_1$ lies in the interval $[0, 1]$ and perfectly characterizes the maximum bias with which any physical measurement can distinguish the two states.

A quantum measurement is described by a Positive Operator-Valued Measure (POVM), which is a set of positive semidefinite operators $\{E_i\}$ such that $\sum_i E_i = I$. Operationally, if a system is in state $\rho$, the exact probability of obtaining outcome $i$ upon applying this measurement is precisely $\operatorname{Tr}(E_i \rho)$.

An \emph{effect} is an operator $E$ with $0\preceq E\preceq I$, such as the POVM element representing acceptance. For density operators $\alpha,\beta$, the Holevo--Helstrom theorem~\cite{watrous2018} gives
\begin{equation}\label{eq:xeq-effect}
    \left|\operatorname{Tr}(E\alpha)-\operatorname{Tr}(E\beta)\right|
    \le\|\alpha-\beta\|_{\mathrm{tr}}.
\end{equation}
For equally likely $\alpha$ and $\beta$, optimal binary discrimination succeeds with probability $(1+\|\alpha-\beta\|_{\mathrm{tr}})/2$.

For a comprehensive introduction to general quantum information theory, we refer the reader to standard texts such as Nielsen and Chuang~\cite{nielsen2000} and Watrous~\cite{watrous2018}.

\subsection{The \texorpdfstring{$(p, \eta)$}{(p, eta)}-Predictor and State Identification}
\label{sec:predictors}

To analyze protocols in the bounded-error regime, we utilize the framework of bounded-error quantum state identification introduced by Gavinsky, Kempe, Regev, and de Wolf~\cite{gavinsky2006}.

Consider a quantum state $\alpha_X$ that encodes a uniformly distributed random bit $X \in \{0,1\}$. A referee wishes to guess $X$ but is allowed to occasionally output a conclusive failure symbol,~$\bot$, representing ignorance.
A POVM with three outcomes $\{0, 1, \bot\}$ is defined as a \emph{$(p, \eta)$-predictor} for $X$ if it satisfies two conditions:
\begin{enumerate}
    \item It produces a definitive guess with probability exactly $p = \Pr[\text{outcome} \neq \bot]$.
    \item Conditioned on making a definitive guess, the probability that the guess is incorrect is at most $\eta$. That is, $\Pr[X \neq \text{outcome} \mid \text{outcome} \neq \bot] \le \eta$.
\end{enumerate}
The same definition applies to any finite-valued random variable, with conclusive outcomes given by its possible values and one additional abort outcome.
The quantity $D_\eta(\alpha_0, \alpha_1)$ is defined as the supremum over all possible POVMs of the success probability $p$ such that the measurement is a $(p, \eta)$-predictor for $X$. This metric cleanly captures the trade-off between the frequency of a measurement producing a conclusive answer and the reliability of that answer.

\subsection{Communication Models and Error Regimes}
\label{sec:smpmodels}

In the multiparty Simultaneous Message Passing (SMP) model, $k$ spatially separated players $P_1, \dots, P_k$ receive private inputs $x_j \in X_j$. They cannot communicate with each other; instead, each player applies a local encoding function and sends a single message to a central referee. The referee, who receives no input, must process these $k$ messages to evaluate a relational problem $R \subseteq X_1 \times \cdots \times X_k \times Z$, outputting a solution $z \in Z$. If $m_j$ denotes the worst-case length of player $P_j$'s message, we measure communication throughout the paper by the maximum message length
\[
    m:=\max_{j\in[k]}m_j.
\]
Thus, a lower bound on $m$ means that at least one player must send a message of that length. When total communication is needed, we write $C:=\sum_{j=1}^k m_j$; in particular, $C\le km$. Unless explicitly stated otherwise, communication complexity refers to the maximum message length.

We evaluate communication costs across three principal resources:
\begin{itemize}
  \item \textbf{Classical Public-Coin ($R^{\parallel}_{\mathrm{pub}}$):} All players (excluding the referee) share access to a common, unbounded random string that does not count toward the communication cost. The messages sent to the referee are classical bit strings.
  \item \textbf{Classical Private-Coin ($R^{\parallel}_{\mathrm{priv}}$):} Players do not share any public randomness, but each player has access to an independent, local source of random coin flips. The messages sent to the referee are classical bit strings. This model is weaker than the unentangled quantum model ($Q^{\parallel}$), because an $m$-qubit quantum message can perfectly simulate $m$ classical private coin flips.
  \item \textbf{Quantum without prior entanglement ($Q^{\parallel}$):} Each player $P_j$ sends a quantum state $\rho_j(x_j)$ of dimension $2^{m_j}$ (i.e., an $m_j$-qubit message). The referee applies a POVM to the joint state $\bigotimes_{j=1}^k \rho_j(x_j)$. Crucially, the players do not share any prior entanglement or public coins; the global state received by the referee is strictly a product state over the players' subspaces.
\end{itemize}

We evaluate the success of these protocols under two strict error definitions, governed by a constant parameter $\varepsilon \in [0, 1)$:
\begin{itemize}
  \item \textbf{Bounded-Error:} The referee is allowed to output a special failure symbol $\bot$. The worst-case probability over all inputs that the protocol fails---either by outputting an invalid tuple $(x_1, \dots, x_k, z) \notin R$ or by outputting $\bot$---is at most $\varepsilon$. Equivalently, the referee outputs a valid tuple with probability at least $1 - \varepsilon$.
  \item \textbf{Unambiguous:} The referee is allowed to output a special failure symbol $\bot$ with probability at most $\varepsilon$. However, the protocol is strictly forbidden from making an error: whenever the referee outputs a definitive answer $z \neq \bot$, the tuple must be valid, meaning $(x_1,\dots,x_k,z) \in R$ with probability $1$.
\end{itemize}
Because unambiguous protocols never output an incorrect answer and only fail by aborting, any valid unambiguous protocol is trivially also a valid bounded-error protocol. Note that while upper bounds must satisfy these criteria in the worst-case, lower bounds are proven by analyzing the average-case error over a specifically chosen hard input distribution.
\subsection{Information-Theoretic Tools}

To establish lower bounds on quantum communication, we utilize the Random Access Code (RAC) lemma, originally pioneered by Nayak~\cite{nayak1999} and adapted for bounded-error state identification by Gavinsky et al.~\cite{gavinsky2006}. The RAC lemma bounds the sum of the capacities of individual bit predictors extractable from a single quantum state, and applies uniformly to both the unambiguous and bounded-error regimes. Here $H(p)=-p\log_2p-(1-p)\log_2(1-p)$ is the binary entropy, with its continuous values at $p=0,1$.

\begin{proposition}[Random Access Codes~\cite{gavinsky2006, nayak1999}] \label{thm:random-access}
    Let $X=X_1\cdots X_n$ be uniformly distributed on $\{0,1\}^n$, and let $Q$ be an $m$-qubit encoding of $X$. Suppose that, for every $i\in[n]$, there exists a measurement on $Q$ whose outcome is a $(\lambda_i,\eta_i)$-predictor for $X_i$, where $0\le \eta_i\le 1/2$. The measurement may depend on $i$. Then
    \[
        \sum_{i=1}^n
        \lambda_i\bigl(1-H(\eta_i)\bigr)
        \le m.
    \]

    Consequently, this lemma yields three critical geometric constraints:
    \begin{enumerate}
        \item \textbf{Zero-Error $L_1$ Bound:} If every predictor is unambiguous, then $\eta_i=0$ for all $i$. Since $H(0)=0$, the RAC inequality takes the simple form $\sum_{i=1}^n \lambda_i \le m$.
        \item \textbf{Fixed-Error $L_1$ Bound:} If the predictors operate with a fixed error $\eta < 1/2$, then the sum of the success probabilities is bounded linearly by the encoding size: $\sum_{i=1}^n \lambda_i = O_\eta(m)$.
        \item \textbf{$L_2$ Trace-Distance Bound:} If the predictors are derived strictly from distinguishing two mixed states $\rho_i^0$ and $\rho_i^1$, operating as $(1, 1/2 - d_i/2)$-predictors where $d_i = \|\rho_i^0 - \rho_i^1\|_{\mathrm{tr}}$, then the sum of the squared trace distances is bounded by the encoding size: $\sum_{i=1}^n d_i^2 = O(m)$.
    \end{enumerate}
\end{proposition}

For the explicit constants in \Cref{sec:xeq}, we use the standard binary-entropy estimate appearing in the proof of \cite[Claim~3.1]{gavinsky2013}:
\begin{equation}\label{eq:xeq-entropy-bound}
    1-H\!\left(\frac{1-u}{2}\right)\ge\frac{u^2}{2\ln2},
    \qquad 0\le u\le1.
\end{equation}

Finally, to pool classical capacities and trace distances across multiple players, we repeatedly use the standard Arithmetic Mean-Geometric Mean inequality in the following form
\begin{proposition}[AM-GM inequality~\cite{Veljan2017}]
  Let $x_1, \dots, x_N$ be non-negative real numbers, it holds that:
\[
    \prod_{j=1}^N x_j \le \left( \frac{1}{N} \sum_{j=1}^N x_j \right)^N.
\]  
\end{proposition}

\subsection{Parity State Factorization}

Both lower-bound arguments use the following factorization of the distinguishability of parity-conditioned product states.

\begin{lemma}[Parity State Factorization {\cite[Claim~3.2]{gavinsky2013}}]\label{lem:parity-factorization}
Let $n \in \mathbb{N}$ and let $X_1,\dots,X_n$ be independent, uniformly distributed random bits. For each $j \in [n]$ and $x \in \{0,1\}$, let $\rho_j^x \in \mathrm{D}(\mathcal{H}_j)$ be the state of the $j$-th quantum system conditioned on $X_j=x$. Define the parity variable $U_n=\bigoplus_{j=1}^n X_j$, and let $\alpha_n^u$ be the global state conditioned on $U_n=u$:
\[
    \alpha_n^u
    =
    \frac{1}{2^{n-1}}
    \sum_{\substack{x\in\{0,1\}^n\\ \bigoplus_{j=1}^n x_j=u}}
    \bigotimes_{j=1}^n \rho_j^{x_j},
    \qquad u\in\{0,1\}.
\]
Then the trace distance between the two parity states factors perfectly:
\[
    \|\alpha_n^0 - \alpha_n^1\|_{\mathrm{tr}}
    =
    \prod_{j=1}^{n} \|\rho_j^0 - \rho_j^1\|_{\mathrm{tr}}.
\]
\end{lemma}

\section{A Total Boolean Function Separation}
\label{sec:xeq}

We first establish the advantage of public randomness for a total Boolean function. We study \textsc{XOR Equality}, denoted $\operatorname{XEQ}_{k,n}$, which asks whether the bitwise XOR of the players' input strings is zero. This extends ordinary two-player Equality to a decision problem defined on unrestricted inputs. Public randomness makes the problem easy by allowing all players to apply the same random linear test, whose outcomes combine into a test of their global XOR. To understand the limitations of independent quantum messages, we analyze conditional average messages for parities of subsets of coordinates. These two perspectives yield a bounded-error separation: constant public-coin maximum message length versus a polynomial quantum lower bound when the number of players grows.

\begin{definition}[Problem $\operatorname{XEQ}_{k,n}$]\label{def:xeq}
For integers $k\ge2$ and $n\ge1$, player $P_j$ receives an unrestricted string $x_j\in\mathbb F_2^n$, and the referee must output
\[
 \operatorname{XEQ}_{k,n}(x_1,\ldots,x_k)
 :=\begin{cases}
    1, & \text{if $\bigoplus_{j=1}^k x_j=0^n$}\\
    0, & \text{if not}
 \end{cases}
\]
This is a total Boolean function, and for $k=2$, it is the usual Equality function.
\end{definition}

\subsection{Classical Public-Coin Upper Bound}

\begin{proposition}[Public-coin protocol for XOR Equality]\label{prop:xeq-public}
For every positive integer $t$, there is a public-coin SMP protocol for $\operatorname{XEQ}_{k,n}$ using $t$ bits per player. It accepts every positive input and accepts each negative input with probability $2^{-t}$. In particular, two bits per player suffice for worst-case error at most $1/4$.
\end{proposition}
\begin{proof}
The players use their shared randomness to choose independent uniform vectors $a_1,\ldots,a_t\in\mathbb F_2^n$. Player $P_j$ sends the single $t$-bit message
\[
 \bigl(\langle a_1,x_j\rangle,\ldots,\langle a_t,x_j\rangle\bigr).
\]
For each $r\in[t]$, the referee XORs the $r$-th received bits and accepts if and only if all $t$ resulting bits are zero. Writing $v=\bigoplus_jx_j$, the result of test $r$ is
\[
 \bigoplus_{j=1}^k\langle a_r,x_j\rangle=\langle a_r,v\rangle.
\]
The referee does not need to know $a_r$, so the protocol respects our convention that only the players share the public randomness.

If $v=0$, all tests accept. If $v\ne0$, choose a coordinate $i$ with $v_i=1$. Conditional on every coordinate of $a_r$ except $(a_r)_i$, the value $\langle a_r,v\rangle$ is a fixed bit XORed with the independent uniform bit $(a_r)_i$. Each test therefore accepts with probability $1/2$. Independence of the $t$ vectors gives acceptance probability $2^{-t}$. Taking $t=2$ completes the proof.
\end{proof}

\subsection{Conditional Average Messages for Parities}
\label{sec:xeq-tools}

Fix a message encoding $\rho:\mathbb F_2^n\to\mathrm D(\mathcal H)$ with $\dim\mathcal H\le2^m$, and let $X$ be uniform on $\mathbb F_2^n$. For $a\ne0$ and $b\in\{0,1\}$, define
\begin{equation}\label{eq:xeq-average-states}
 \rho_a^b
 :=\mathbb E[\rho(X)\mid\langle a,X\rangle=b]
 =2^{1-n}\sum_{x:\langle a,x\rangle=b}\rho(x),
 \qquad
 d(a):=\|\rho_a^0-\rho_a^1\|_{\mathrm{tr}}.
\end{equation}
The normalization uses the fact that both classes of a nonzero parity contain $2^{n-1}$ inputs. The Holevo--Helstrom measurement is a $(1,(1-d(a))/2)$-predictor for $\langle a,X\rangle$. Thus $d(a)$ quantifies the information about this parity available in the message under a uniform local input. The measurement may depend on $a$; simultaneous recovery of several parities is not assumed. For $a=e_i$, these are precisely the coordinate-conditioned average messages used later in the Index Coordination proof. For player $j$, we write $\rho_{j,a}^b$ and $d_{j,a}$.

\begin{lemma}[Span of Predictable Parities]\label{lem:xeq-span}
Let
\[
 B=\{a\in\mathbb F_2^n\setminus\{0\}:d(a)>1/2\},
 \qquad W=\operatorname{span}_{\mathbb F_2} B.
\]
Then $\dim W\le8(\ln2)m$.
\end{lemma}
\begin{proof}
If $W=\{0\}$ there is nothing to prove. Otherwise choose a basis $a_1,\ldots,a_\ell$ of $W$ from $B$. The linear map
\[
 x\longmapsto\bigl(\langle a_1,x\rangle,\ldots,\langle a_\ell,x\rangle\bigr)
\]
has rank $\ell$. Consequently, $U=(\langle a_1,X\rangle,\ldots,\langle a_\ell,X\rangle)$ is uniform on $\mathbb F_2^\ell$, so its coordinates are independent. The average states $\mathbb E[\rho(X)\mid U=u]$ form an $m$-qubit encoding of $U$. Its conditional states for the $t$-th bit are exactly $\rho_{a_t}^0$ and $\rho_{a_t}^1$.

Apply \Cref{thm:random-access} to these $\ell$ predictors, with conclusive probability one and error $(1-d(a_t))/2$. Using~\eqref{eq:xeq-entropy-bound},
\[
 m\ge\sum_{t=1}^{\ell}\left[1-H\!\left(\frac{1-d(a_t)}2\right)\right]
 \ge\frac{1}{2\ln2}\sum_{t=1}^{\ell}d(a_t)^2
 \ge\frac{\ell}{8\ln2}.
\]
This proves the dimension bound.
\end{proof}

The preceding argument controls the dimension of the span of well-predicted parities, rather than their number. An aggregate estimate over all parities requires a different bound. We use the following consequence of the matrix-valued norm estimates of Ben-Aroya, Regev, and de Wolf.

\begin{proposition}[Aggregate Parity Distinguishability
{\cite[proof of Lemma~6]{benaroya2008hypercontractive}}]
\label{prop:xeq-aggregate}
For an encoding $\rho$ and parity distinguishabilities $d(a)$
defined above,
\begin{equation}\label{eq:xeq-aggregate}
    \sum_{a\in\mathbb F_2^n\setminus\{0\}} d(a)^2 \le 2^m.
\end{equation}
\end{proposition}
This bound follows by taking the Schatten exponent equal to two in the proof of \cite[Lemma~6]{benaroya2008hypercontractive}, with the trace norm translated to our convention.

Next, let $X_1,\ldots,X_k$ be independent uniform strings and set $V=\bigoplus_jX_j$. Define the joint average messages
\begin{equation}\label{eq:xeq-sigma}
 \Sigma_v
 :=\mathbb E\!\left[\bigotimes_{j=1}^k\rho_j(X_j)\,\middle|\,V=v\right],
 \qquad v\in\mathbb F_2^n.
\end{equation}
Although the message on each fixed input tuple is a product state, conditioning on $V=v$ correlates the inputs, and $\Sigma_v$ need not be a product state.

\begin{lemma}[Global XOR Conditioning]\label{lem:xeq-global}
For every nonzero $z\in\mathbb F_2^n$,
\begin{equation}\label{eq:xeq-global-distance}
 \|\Sigma_0-\Sigma_z\|_{\mathrm{tr}}
 \le\sum_{a:\langle a,z\rangle=1}\prod_{j=1}^k d_{j,a}.
\end{equation}
\end{lemma}
\begin{proof}
For $a\ne0$, let $\Gamma_a^b$ be the average joint message conditioned on $\langle a,V\rangle=b$. The bits $\langle a,X_j\rangle$ are independent and uniform, and the average message conditioned on their full tuple is the tensor product of the corresponding $\rho_{j,a}^{b_j}$. Hence \Cref{lem:parity-factorization} gives
\begin{equation}\label{eq:xeq-global-parity}
 \|\Gamma_a^0-\Gamma_a^1\|_{\mathrm{tr}}
 =\prod_{j=1}^k d_{j,a}.
\end{equation}

Since $V$ is uniform, averaging over its two parity classes yields
\[
 \Gamma_a^0-\Gamma_a^1
 =2^{1-n}\sum_{v\in\mathbb F_2^n}(-1)^{\langle a,v\rangle}\Sigma_v.
\]
We claim that
\begin{equation}\label{eq:xeq-global-identity}
 \Sigma_0-\Sigma_z
 =\sum_{a:\langle a,z\rangle=1}(\Gamma_a^0-\Gamma_a^1).
\end{equation}

To verify the claim, use the elementary identity
\[
 \sum_{a\in\mathbb F_2^n}(-1)^{\langle a,w\rangle}
 =2^n\mathbf1[w=0].
\]
For $w\ne0$, the terms cancel in pairs $a,a\oplus e_i$ for any coordinate with $w_i=1$; for $w=0$, every term is one. It follows that
\begin{align*}
 \sum_{a:\langle a,z\rangle=1}(-1)^{\langle a,v\rangle}
 &=\frac12\sum_a\left((-1)^{\langle a,v\rangle}
                     -(-1)^{\langle a,v\oplus z\rangle}\right)\\
 &=2^{n-1}\bigl(\mathbf1[v=0]-\mathbf1[v=z]\bigr).
\end{align*}
Substitution into the sum of the expansions of $\Gamma_a^0-\Gamma_a^1$ proves~\eqref{eq:xeq-global-identity}. The triangle inequality and~\eqref{eq:xeq-global-parity} now give~\eqref{eq:xeq-global-distance}.
\end{proof}

\subsection{Quantum Lower Bound}

\begin{theorem}[Quantum Lower Bound for XOR Equality]\label{thm:xeq-lower}
For $n\ge1$ and $k\ge8$, every quantum SMP protocol for $\operatorname{XEQ}_{k,n}$ without shared randomness or prior entanglement and with worst-case error at most $1/3$ satisfies
\begin{equation}\label{eq:xeq-main-bound}
 C:=\sum_{j=1}^k m_j\ge\frac18\min\{n,k^2\},
 \qquad
 m:=\max_{j\in[k]}m_j\ge\frac18\min\{n/k,k\}.
\end{equation}
\end{theorem}
\begin{proof}
Fix a protocol and suppose, for a contradiction, that $C<\frac18\min\{n,k^2\}$. The local encoding of player $P_j$ is denoted by $\rho_j$. We use uniform local inputs to define $\rho_{j,a}^b$ and $d_{j,a}$ as above.

\noindent{\bf Bounding the Information.}
For every player, set
\[
 W_j=\operatorname{span}_{\mathbb F_2}\{a\ne0:d_{j,a}>1/2\},
 \qquad W=W_1+\cdots+W_k.
\]
By \Cref{lem:xeq-span} and the assumption $C<n/8$,
\[
 \dim W\le\sum_j\dim W_j\le8(\ln2)C<n\ln2<n.
\]

\noindent{\bf Constructing the Hard Distribution.}
Because the standard bilinear form on $\mathbb F_2^n$ is nondegenerate, there is a nonzero $z\in W^\perp$. Fix one such $z$. For every $a$ satisfying $\langle a,z\rangle=1$,
\begin{equation}\label{eq:xeq-weak-parities}
 d_{j,a}\le1/2\qquad\text{for all }j\in[k].
\end{equation}
Otherwise $a$ would belong to some $W_j\subseteq W$, contradicting $z\in W^\perp$.

Consider the uniform distributions on input tuples with global XOR zero and $z$, respectively. Their average messages are $\Sigma_0$ and $\Sigma_z$. All tuples in the first distribution are positive, and all tuples in the second are negative. These distributions may depend on the fixed protocol, since correctness is required on every input.

\noindent{\bf Bounding the Referee's Distinguishability.}
Relabel the players so that $$m_1\le m_2\le\cdots\le m_k.$$ Applying \Cref{lem:xeq-global}, then~\eqref{eq:xeq-weak-parities} to $k-2$ factors, gives
\begin{align}
 \|\Sigma_0-\Sigma_z\|_{\mathrm{tr}}
 &\le\sum_{a:\langle a,z\rangle=1}\prod_{j=1}^k d_{j,a}\nonumber\\
 &\le2^{-(k-2)}\sum_{a:\langle a,z\rangle=1}d_{1,a}d_{2,a}\nonumber\\
 &\le2^{-(k-2)}
       \left(\sum_{a\ne0}d_{1,a}^2\right)^{1/2}
       \left(\sum_{a\ne0}d_{2,a}^2\right)^{1/2}\nonumber\\
 &\le2^{2-k+(m_1+m_2)/2}.\label{eq:xeq-distance-bound}
\end{align}
The third line is Cauchy--Schwarz, after enlarging the summation set, and the last line uses \Cref{prop:xeq-aggregate}. Since $m_1,m_2$ are the two smallest message lengths, their average is at most $C/k$. Thus $C<k^2/8$ and $k\ge8$ imply
\begin{equation}\label{eq:xeq-small-distance}
 \|\Sigma_0-\Sigma_z\|_{\mathrm{tr}}
 \le2^{2-k+C/k}<2^{2-7k/8}\le\frac1{32}.
\end{equation}

\noindent{\bf Forcing the Communication Bound.}
Let $E$ be the POVM element corresponding to output $1$. This is an effect even if the protocol also permits an abort outcome. Averaging worst-case correctness over the two distributions gives
\[
 \operatorname{Tr}(E\Sigma_0)\ge\frac23,
 \qquad \operatorname{Tr}(E\Sigma_z)\le\frac13.
\]
The acceptance bound~\eqref{eq:xeq-effect} therefore requires $\|\Sigma_0-\Sigma_z\|_{\mathrm{tr}}\ge1/3$, contradicting~\eqref{eq:xeq-small-distance}. This proves the bound on $C$. Finally, $C\le km$ gives the bound on the maximum message length.
\end{proof}

\begin{corollary}[Total Boolean Separation]\label{cor:xeq-separation}
For every integer $k\ge8$, set $n=k^2$. With worst-case error at most $1/3$,
\[
 R^{\parallel}_{\mathrm{pub}}(\operatorname{XEQ}_{k,n})\le2,
 \qquad
 Q^{\parallel}(\operatorname{XEQ}_{k,n})\ge\frac{k}{8}=\frac{\sqrt n}{8}.
\]
\end{corollary}
\begin{proof}
Apply \Cref{prop:xeq-public,thm:xeq-lower} with $t=2$ and $n=k^2$.
\end{proof}

\section{Exponential Separations for Multiparty Index Coordination}
\label{sec:separation}

Having established a polynomial separation for a total Boolean function, we now show that relaxing to relational problems allows an exponential separation between classical shared randomness and unentangled quantum communication in the multiparty SMP model. To achieve this, we study a relational problem that we call \textsc{Multiparty Index Coordination}, denoted $\operatorname{IC}_{k,n}$.

This problem serves as the natural $k$-party generalization of the bipartite \textsc{Index Coordination} problem, which was originally utilized by Gavinsky et al.~\cite{gavinsky2004, gavinsky2006} to demonstrate the incomparability of public coins and quantum messages for two players. In the original two-party setting, Alice and Bob must coordinate to output their respective bits for a specific index subset known only to Bob. In our multiparty extension, the first $k-1$ players each receive an $n$-bit string. The $k$-th player receives an $n$-bit string alongside a selector string $s$ with a Hamming weight of exactly $n/2$, which acts as a mask indicating the valid target indices. The goal of the referee is to successfully output the players' bits for any one of these valid indices.

The intuition behind this problem captures the essence of distributed coordination bottlenecks. In the classical public-coin model, the players can achieve this task by using the shared randomness to collectively sample a random index. Because the target set has a relative frequency of $1/2$, a constant number of independent samples contains a valid index with high probability, allowing the players to succeed using only logarithmic communication.

In contrast, in the unentangled quantum model without shared randomness, the players lack any mechanism to coordinate. The first $k-1$ players have no knowledge of the selector string $s$, meaning they do not know which parts of their own input to send for the referee to succeed. Consequently, their independent quantum messages must inherently carry enough predictive power to answer queries about a large fraction of their inputs simultaneously.

We formalize this problem as follows:

\begin{definition}[Problem $\operatorname{IC}_{k,n}$]
    The problem involves $k\ge2$ players $P_1, \dots, P_k$ and a central referee, where $n\ge2$ is even.
    \begin{itemize}
        \item \textbf{Input:} Each player $P_j$ for $j \in \{1, \dots, k-1\}$ receives a string $x_j \in \{0,1\}^n$.
        Player $P_k$ receives a string $x_k \in \{0,1\}^n$ and a selector string $s \in \{0,1\}^n$ with Hamming weight exactly $|s| = n/2$.
        \item \textbf{Output:} The referee must output a tuple $(i, x_1^{(i)}, \dots, x_k^{(i)})$ for some index $i$ such that $s^{(i)} = 1$.
    \end{itemize}
\end{definition}

\subsection{Classical Protocols for \texorpdfstring{$\operatorname{IC}_{k,n}$}{IC\_(k,n)}}

\noindent{\bf Public-Coin Upper Bound.}
This problem admits a highly efficient classical protocol utilizing shared randomness. The players use the public coin to sample an index $i \in [n]$ uniformly at random. Since the target set has relative frequency $1/2$, a random index is valid with probability $1/2$.
All players send their bit $x_j^{(i)}$ and the index $i$; player $P_k$ additionally sends $s^{(i)}$. If $s^{(i)}=0$, the referee aborts ($\bot$). Otherwise, the referee outputs the tuple.

Crucially, this protocol is strictly \emph{unambiguous}: it never outputs an incorrect tuple. To bound the abort probability below a constant $\varepsilon\in(0,1)$, the players simply repeat this process $O(\log(1/\varepsilon))$ times. Its maximum message length is therefore bounded by
\[
    R^{\parallel}_{\mathrm{pub}}(\operatorname{IC}_{k,n}) = O(\log n).
\]

\noindent{\bf Private-Coin Baseline.}
Without public randomness, a simple strategy generalizes the protocol described by Ambainis~\cite{ambainis1996}.
Assume first that $n^{1/k}$ is an integer. Each player treats their $n$-bit input as a flattened $k$-dimensional tensor of size $n^{1/k} \times \dots \times n^{1/k}$. Each independently selects a uniform random coordinate $c_j \in [n^{1/k}]$ for the $j$-th dimension and sends $c_j$ together with the corresponding sub-tensor of dimension $k-1$. Additionally, player $P_k$ sends the corresponding hyperplane of the selector string $s$. Geometrically, these sub-tensors are orthogonal hyperplanes that intersect in exactly one common index $i$ (see \Cref{fig:priv_prot}). The referee checks whether $s^{(i)}=1$, outputting the corresponding tuple if so and aborting otherwise. The intersection is uniformly distributed over the $n$ coordinates, so each trial succeeds with probability $1/2$.

\vspace{5mm}

\begin{figure}[htbp]
    \centering
    \begin{subfigure}[b]{0.4\textwidth}
        \centering
        \includegraphics[width=\linewidth]{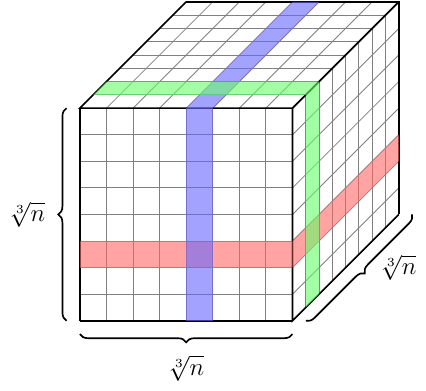}
        \caption{Input viewed as a cube when $k=3$. Each player selects a color and sends the corresponding slice.}
        \label{fig:first}
    \end{subfigure}%
    \hfill
    \begin{subfigure}[b]{0.305\textwidth}
        \centering
        \includegraphics[width=\linewidth]{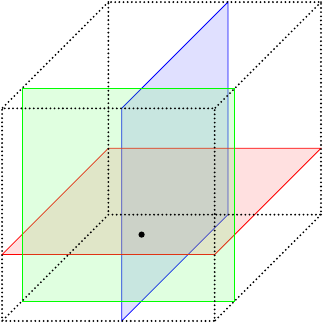}
        \vspace{4mm}
        \caption{\raggedright Intersection of the three players' planes at a single coordinate.}
        \label{fig:second}
    \end{subfigure}
    \caption{Private-coin protocol for Multiparty Index Coordination.}
    \label{fig:priv_prot}
\end{figure}

For arbitrary $n$, we use a small amount of zero padding and allow slightly different side lengths. Specifically, pad the inputs and selector to length $N=2^L$, where $L=\lceil\log_2 n\rceil$, and distribute the $L$ binary index coordinates as evenly as possible among the $k$ dimensions. Each side then has length $2^{\lfloor L/k\rfloor}$ or $2^{\lceil L/k\rceil}$, with total volume $N<2n$; sides of length one are allowed. The arrangement is fixed and known to all parties. Since padded selector positions are zero, they only cause the referee to abort. Each trial still succeeds with probability $n/(2N)\ge1/4$, and every hyperplane contains at most $2N^{1-1/k}\le4n^{1-1/k}$ entries.

Thus, a constant number of independent parallel repetitions suffices for any fixed abort allowance $\varepsilon\in(0,1)$. The maximum message length is proportional to the size of one hyperplane, with an additional constant factor for the selector. The $O(\log n)$ coordinate headers are absorbed by this bound for $k\ge2$. Consequently,
\[
    R^{\parallel}_{\mathrm{priv}}(\operatorname{IC}_{k,n})
    = O\left(n^{1-1/k}\right).
\]
The protocol is unambiguous, since every conclusive output is valid.

\subsection{Bounded-Error Quantum Lower Bound}
\label{sec:bounded-error}

We first establish the relational separation in the general bounded-error regime. Deriving a lower bound by generalizing the techniques used for the two-party problem is blocked by the complexities of semidefinite programming (SDP) duality~\cite{gavinsky2006}.

To circumvent this, we adapt our techniques to utilize the original two-register asymmetric direct product theorem of Gavinsky et al.~\cite{gavinsky2006}. We conceptually group the first $k-1$ independent players into a single ``meta-player'' by studying the parity of their inputs. This workaround allows us to embed the $k$-register message into a simple 2-register bipartite problem, enabling the use of these established bounded-error results.

The parity-state factorization lemma from the preliminaries allows us to replace the massive global multiparty state property with the product of local capacities, linking our setup to Gavinsky et al.'s two-register theorem:

\begin{proposition}[Bounded-Error 2-Register Direct Product
{\cite[Corollary~1]{gavinsky2006}}]\label{prop:gavinsky-dpt}
Let $0\le\eta<1/2$, and let $\alpha_X\otimes\beta_Y$ encode independent, uniformly distributed random bits $X,Y$. Let $b=D_\eta(\beta_0,\beta_1)$ be the maximum conclusive probability of an $\eta$-error predictor for $Y$. The optimal conclusive probability $p$ of a $(p,\eta/2)$-predictor for $(X,Y)$ satisfies
\[
    p\le32\|\alpha_0-\alpha_1\|_{\mathrm{tr}}\cdot b.
\]
\end{proposition}

Now we are ready to prove the communication lower bound for the bounded-error regime.

\begin{theorem}[Bounded-Error Lower Bound]\label{thm:bounded-error}
For any constant error $\varepsilon < 1/8$, every bounded-error quantum SMP protocol without shared randomness or pre-shared entanglement that solves $\operatorname{IC}_{k,n}$ has maximum message length
\[
    m=\max_{j\in[k]}m_j=\Omega\left(n^{\frac{k-1}{k+1}}\right)
\]
qubits. Equivalently, at least one player must send this many qubits. The implicit constant depends only on $\varepsilon$ and is independent of both $n$ and $k$.
\end{theorem}

\begin{proof}[Proof sketch.]
    Fix a protocol in which every message has at most $m$ qubits. For each of the first $k-1$ players, the random access code bound limits the total squared distinguishability of the individual input bits. Pooling these quantities lets us choose a selector set of size $n/2$ on which all those players are collectively weak. For an index in this set, we replace their full tuple of bits by its parity: any referee that recovers the tuple also recovers this parity together with the last player's bit. Independence makes the distinguishability of the two parity states equal to the product of the local distinguishabilities, which can be bounded by AM--GM. The two-register identification bound (\Cref{prop:gavinsky-dpt}) then bounds the probability of jointly predicting the parity and the last bit by this small product times the last player's prediction capacity. A second random access code bound over the information of the $k$-th player bounds the sum of those capacities over all indices. Summing the resulting success bounds and forcing the constant overall success yields $m=\Omega(n^{(k-1)/(k+1)})$.
\end{proof}

\begin{proof}
Fix a bounded-error protocol \(P\) in which every player sends at most
\(m\) qubits. We assume \(k\ge 2\) and that \(n\) is even.
Throughout the proof, the input strings are independent and uniformly distributed.

\noindent{\bf Bounding the Information.}
For every player \(j\in[k-1]\) and coordinate \(i\in[n]\), let
\(\rho_{j,i}^0\) and \(\rho_{j,i}^1\) be the average message states
conditioned on \(x_{j,i}=0\) and \(x_{j,i}=1\), respectively, and
write \(d_{j,i}:=\|\rho_{j,i}^0-\rho_{j,i}^1\|_{\mathrm{tr}}\).
The Holevo--Helstrom measurement is a
\((1,1/2-d_{j,i}/2)\)-predictor for \(x_{j,i}\). Since
\(1-H(1/2-t/2)\ge c_0t^2\) for some universal constant \(c_0>0\)
and every \(t\in[0,1]\), the RAC inequality implies
\[
    \sum_{i=1}^n d_{j,i}^2\le C_0m
    \qquad\text{for every }j\in[k-1],
\]
where \(C_0\ge1\) is a universal constant.

We pool these quantities by defining
\(c_i:=\sum_{j=1}^{k-1}d_{j,i}^2\). Summing the preceding RAC bounds
over the first \(k-1\) players gives
\[
    \sum_{i=1}^n c_i
    =
    \sum_{j=1}^{k-1}\sum_{i=1}^n d_{j,i}^2
    \le C_0(k-1)m.
\]

\noindent{\bf Constructing the Hard Distribution.}
We now use the pooled quantities to fix the selector of player \(P_k\).
By Markov's inequality, fewer than \(n/2\) coordinates satisfy
\[
    c_i>
    2\frac{C_0(k-1)m}{n}
    =
    \frac{2C_0(k-1)m}{n}.
\]
Consequently, there is a subset \(S\subseteq[n]\) of size exactly
\(n/2\) such that \(c_i\le 2C_0(k-1)m/n\) for every \(i\in S\). We fix the
selector string \(s\) to be the indicator of this set.

With \(s\) fixed, let
\[
    \beta_{i,s}^z
    :=
    \mathbb{E}_{x_k:x_{k,i}=z}\bigl[\rho_k(x_k,s)\bigr],
    \qquad
    b_i
    :=
    D_{4\varepsilon}
    \left(\beta_{i,s}^0,\beta_{i,s}^1\right).
\]
Thus \(b_i\) is the optimal conclusive probability of a
\(4\varepsilon\)-error predictor for the \(i\)-th bit of \(x_k\).
Because \(4\varepsilon<1/2\), the RAC inequality applied to the fixed
encoding \(x_k\mapsto\rho_k(x_k,s)\) yields
\[
    \sum_{i=1}^n b_i\le C_1m,
\]
where \(C_1>0\) depends only on the fixed error parameter
\(\varepsilon\), and not on \(n\) or \(k\).

\noindent{\bf Collapsing the First \(k-1\) Players.}
Fix \(i\in S\) and define the parity
\(U_i:=x_{1,i}\oplus\cdots\oplus x_{k-1,i}\). Let
\(\alpha_i^0\) and \(\alpha_i^1\) be the joint message states of the
first \(k-1\) players conditioned on \(U_i=0\) and \(U_i=1\),
respectively. By \Cref{lem:parity-factorization} and AM--GM inequality,
\[
\begin{aligned}
    \left\|\alpha_i^0-\alpha_i^1\right\|_{\mathrm{tr}}
    &=
    \prod_{j=1}^{k-1}d_{j,i} \\
    &\le
    \left(
        \frac{1}{k-1}
        \sum_{j=1}^{k-1}d_{j,i}^2
    \right)^{\frac{k-1}{2}} \\
    &=
    \left(\frac{c_i}{k-1}\right)^{\frac{k-1}{2}}
    \le
    \left(
        2C_0\frac{m}{n}
    \right)^{\frac{k-1}{2}}.
\end{aligned}
\]

\noindent{\bf Bounding the Referee's Success.} Let \(p_i\) be the probability that the referee produces a conclusive
output with index \(i\). Since the fixed selector is supported on
\(S\), every valid output must use an index in \(S\). Call
\(i\in S\) \emph{good} if, conditioned on producing an output with
index \(i\), the reported tuple is correct with probability at least
\(1-2\varepsilon\), and let \(G\subseteq S\) be the set of good
indices.

Because the protocol is correct with probability at least
\(1-\varepsilon\), an averaging argument gives
\[
\begin{aligned}
    1-\varepsilon
    &\le
    \sum_{i\in G}p_i
    +(1-2\varepsilon)\sum_{i\in S\setminus G}p_i \\
    &\le
    1-2\varepsilon+2\varepsilon\sum_{i\in G}p_i.
\end{aligned}
\]
For $\varepsilon>0$, the preceding inequality gives
$\sum_{i\in G}p_i\ge1/2$. For $\varepsilon=0$, the referee
always produces a valid output, so every index with positive
output probability belongs to $G$ and $\sum_{i\in G}p_i=1$.
Hence $\sum_{i\in G}p_i\ge1/2$ in either case.

For every \(i\in G\), run the referee's measurement and return the
parity of the first \(k-1\) reported bits together with the reported
\(k\)-th bit whenever the output index is \(i\); otherwise return
\(\bot\). This gives a \((p_i,2\varepsilon)\)-predictor for
\((U_i,x_{k,i})\). Applying the asymmetric two-register direct-product
bound from \Cref{prop:gavinsky-dpt}, followed by the parity estimate
above, gives
\[
    p_i
    \le
    32
    \left\|\alpha_i^0-\alpha_i^1\right\|_{\mathrm{tr}}b_i
    \le
    32
    \left(
        2C_0\frac{m}{n}
    \right)^{\frac{k-1}{2}}
    b_i.
\]
Summing over the good indices and using the RAC bound
\(\sum_i b_i\le C_1m\), we conclude that
\[
    \frac12
    \le
    \sum_{i\in G}p_i
    \le
    32C_1m
    \left(
        2C_0\frac{m}{n}
    \right)^{\frac{k-1}{2}}.
\]

\noindent{\bf Wrapping Up the Lower Bound.}
Rearranging, with $C_2:=\min\{1,1/(64C_1)\}$, gives
\[
\begin{aligned}
    m
    &\ge C_2^{2/(k+1)}
    (2C_0)^{-(k-1)/(k+1)}
    n^{(k-1)/(k+1)} &\ge \frac{C_2^{2/3}}{2C_0}
    n^{(k-1)/(k+1)}.
\end{aligned}
\]
Here we used $k\ge2$, $C_2\le1$, and $C_0\ge1$.
\end{proof}

\subsection{A Technical Direct-Product Theorem}
\label{sec:qsd-main}

In the multiparty SMP model, the referee receives a composite message in the form of a product state $\rho = \bigotimes_{j=1}^k \rho_j$. To solve the relational problem $\operatorname{IC}_{k,n}$ presented in this paper, the referee must identify which global label $b \in \Sigma_1 \times \cdots \times \Sigma_k$ (corresponding to the players' distributed inputs) is encoded within this product state. The following theorem is a technical tool for proving the unambiguous quantum lower bound: it bounds the referee's probability of identifying a tuple of labels from independent messages.

Direct-product theorems of this form are well known in two-party
communication complexity, typically demonstrating that solving
multiple independent instances of a task requires resources that
scale linearly with the number of instances~\cite{klauck2007,sherstov2012}.
Here, the task is the unambiguous identification of a composite
label from a tensor product of mixed states. A closely related
result on unambiguous quantum state discrimination appeared in
work by Gupta et al.~\cite{gupta2024} while this project was in
progress. In this setting, the central question is to characterize
the optimal measurement strategy under the zero-error constraint.

For the two-player setting ($k = 2$), Gavinsky et al.~\cite{gavinsky2004} obtained a lower bound by exploiting a simple factorization of the supports of the relevant quantum states. For an arbitrary number of players $k$ and arbitrary finite alphabets $\Sigma_j$, we demonstrate that, despite this complexity, the global subspace relevant for unambiguous identification still factorizes cleanly into a tensor product of local exclusion subspaces. This structural property allows us to formally state our direct-product theorem, bounding the global success probability strictly by the product of the local success probabilities.

\begin{theorem}[Unambiguous Product Identification over Finite Alphabets]\label{thm:product-bound-gen}
For each $j \in [k]$, let $\{\alpha_j^{(a)}\}_{a \in \Sigma_j}$ be an ensemble of mixed states on $\mathcal{H}_j$. Define the optimal single-register unambiguous success probability for identifying state $a$ as:
\[
    p_j^{(a)} := \sup_{M_j} \operatorname{Tr}\left(E_a \, \alpha_j^{(a)}\right),
\]
where the supremum is over all unambiguous POVMs $M_j = \{E_x\}_{x \in \Sigma_j} \cup \{E_{\bot}\}$.

Consider the ensemble of product states $\rho_b = \bigotimes_{j=1}^k \alpha_j^{(b_j)}$ for $b \in \Sigma_1 \times \cdots \times \Sigma_k$. The maximum probability $P(b)$ of a global unambiguous measurement correctly identifying a specific sequence $b$ is bounded by the product of the local probabilities:
\[
    P(b) \le \prod_{j=1}^k p_j^{(b_j)}.
\]
\end{theorem}
The proof of this theorem is given in \Cref{app:proof-of-dir}.

To apply this theorem to the SMP model where inputs are drawn from a uniform prior, we bridge the conditional identification probability of a specific sequence to the average capacities constrained by the RAC inequality.

\begin{corollary}[Average Capacity Bound]\label{cor:avg-capacity}
For each register $j$, let $s_j$ be the optimal average success probability, under the uniform prior on $\Sigma_j$, of a single unambiguous POVM that may identify any symbol:
\[
    s_j := \sup_{M_j}\frac{1}{|\Sigma_j|}\sum_{a\in\Sigma_j}
    \operatorname{Tr}\left(E_a\alpha_j^{(a)}\right),
\]
where $M_j=\{E_a\}_{a\in\Sigma_j}\cup\{E_\bot\}$ is required to be unambiguous for the entire ensemble. For any specific target sequence $b \in \Sigma_1 \times \cdots \times \Sigma_k$, the global unambiguous success probability $P(b)$ is bounded by:
\[
    P(b) \le \left( \prod_{j=1}^k |\Sigma_j| \right) \prod_{j=1}^k s_j.
\]
For binary alphabets where $|\Sigma_j| = 2$ for all $j$, this reduces to $P(b) \le 2^k \prod_{j=1}^k s_j$.
\end{corollary}

\begin{proof}
Fix a register $j$ and a target symbol $a$. A one-sided POVM attaining, or arbitrarily approaching, $p_j^{(a)}$ in \Cref{thm:product-bound-gen} is also a valid POVM for the average identification problem: retain its element $E_a$, set every other conclusive element to zero, and assign the remaining operator to $E_\bot$. Under the uniform prior, this POVM succeeds with probability arbitrarily close to $p_j^{(a)}/|\Sigma_j|$. Hence
\[
    p_j^{(a)} \le |\Sigma_j|s_j.
\]
Applying this inequality to each factor in \Cref{thm:product-bound-gen} gives the result.
\end{proof}

\subsection{Unambiguous Quantum Lower Bound}

We now consider the unambiguous regime as a special case of independent interest. Given the highly restricted nature of this setting, we are able to extend the technique used in \cite{gavinsky2004} and directly bound the referee's success probability using \Cref{cor:avg-capacity}. This allows us to establish that unentangled quantum messages offer no asymptotic advantage over private coins (uniformly in the number of players), achieving a stronger lower bound.

\begin{theorem}[Unambiguous Lower Bound]\label{thm:improved-lower-bound}
For any fixed constant $\varepsilon\in[0,1)$, every unambiguous quantum SMP protocol for $\operatorname{IC}_{k,n}$ without shared randomness or prior entanglement and with abort probability at most $\varepsilon$ has maximum message length
\[
    m=\max_{j\in[k]}m_j=\Omega(n^{1 - 1/k})
\]
qubits. Equivalently, at least one player must send this many qubits. The implicit constant depends only on $\varepsilon$ and is independent of both $n$ and $k$.
\end{theorem}

\begin{proof}[Proof sketch.]
    Fix an unambiguous protocol with constant abort probability and maximum message length $m$. For each player and coordinate, let the local capacity be the largest probability of identifying that input bit without error. The zero-error case of the random access code bound says that the sum of a player's local capacities is at most the message length. We pool the capacities of the first $k-1$ players and fix the selector set among coordinates where their total capacity is small as the hard distribution. At any selected coordinate, the unambiguous direct-product theorem bounds the probability of identifying the entire tuple by the product of the local capacities. AM--GM allows us to bound the contribution of the first $k-1$ players, while a second random access code bound controls the sum of the last player's capacities. Therefore the protocol's total success rate is at most $2^{2k-1}m^k/n^{k-1}$. Hence constant success probability requires $m=\Omega(n^{1-1/k})$.
\end{proof}
   
\begin{proof}

We assume there exists an unambiguous protocol for the problem $\operatorname{IC}_{k,n}$. Throughout the proof, the input strings are independent and uniformly distributed; the selector is fixed below. Let $m_j \le m$ be the message size in qubits sent by player $P_j$. For each player $P_j$ with $j\in[k-1]$ and coordinate $i\in[n]$, let $\rho_{j,i}^0$ and $\rho_{j,i}^1$ be the average message states conditioned on $x_{j,i}=0$ and $x_{j,i}=1$, respectively. We denote by $a_j^{(i)}$ the optimal average success probability, under the uniform prior on the bit, of a single unambiguous POVM for this binary ensemble, as in \Cref{cor:avg-capacity}.

\noindent{\bf Bounding the Information.}

For every coordinate $i\in[n]$, the optimal unambiguous measurement for player $P_j$, where $j\in[k-1]$, is an $(a_j^{(i)},0)$-predictor: it produces a conclusive answer with probability $a_j^{(i)}$ and, conditioned on doing so, never outputs the wrong bit. Applying the RAC inequality (\Cref{thm:random-access}) with $\lambda_i=a_j^{(i)}$ and $\eta_i=0$ gives
\[
    \sum_{i=1}^n a_j^{(i)}(1-H(0))\le m_j.
\]
Since $H(0)=0$, this takes the simple form
\[
    \sum_{i=1}^n a_j^{(i)}\le m_j.
\]

We aggregate the capacities of the first $k-1$ players. Let $c_i := \sum_{j=1}^{k-1} a_j^{(i)}$ represent their pooled unambiguous prediction capacity for the $i$-th coordinate, and let $M := \sum_{j=1}^{k-1} m_j \le (k-1)m$ be their total communication budget. Summing the individual RAC inequalities yields
\[
    \sum_{i=1}^n c_i
    =
    \sum_{j=1}^{k-1}\sum_{i=1}^n a_j^{(i)}
    \le
    \sum_{j=1}^{k-1}m_j
    =
    M.
\]

\noindent{\bf Constructing the Hard Distribution.}
Using this aggregated capacity, we construct a hard input distribution by fixing the $k$-th player's selector string $s$. Specifically, we require a subset $S$ of size exactly $n/2$ where the pooled capacity $c_i$ is small.
If $M=0$, then every $c_i$ is zero and any subset $S\subseteq[n]$ of size $n/2$ has the required property. Otherwise, by Markov's inequality, the number of indices $i \in [n]$ exceeding a threshold $T$ is strictly bounded by $M/T$. To guarantee there are at least $n/2$ ``good'' indices remaining, we must restrict the maximum number of ``bad'' indices to be strictly less than $n - n/2 = n/2$. Setting $M/T = n/2$ yields the precise threshold $T = 2M/n$. Therefore, there exists a static subset $S \subseteq [n]$ of size $n/2$ such that for all $i \in S$, $c_i \le 2M/n$. We fix the hard instance distribution by setting $s$ to be the indicator vector of this specific set $S$.

With this selector now fixed, let $\rho_{k,i}^0$ and $\rho_{k,i}^1$ be the average message states of player $P_k$ on inputs $(x_k,s)$, conditioned on $x_{k,i}=0$ and $x_{k,i}=1$, respectively. Let $a_k^{(i)}$ be the optimal average success probability of a single unambiguous POVM for this binary ensemble. Applying the RAC inequality to the encoding $x_k\mapsto\rho_k(x_k,s)$ for this fixed selector gives
\[
    \sum_{i=1}^n a_k^{(i)}\le m_k.
\]

\noindent{\bf Bounding the Referee Prediction.}
Let $p_i$ be the unconditional probability, over uniformly random input strings and the protocol's private randomness, that the referee outputs a valid tuple with index $i$. Because the protocol is strictly unambiguous, it is forbidden from outputting an incorrect tuple; moreover, the selector string $s$ is supported on $S$, so a valid answer must use an index in $S$. Consequently, the total probability of generating a valid tuple is $\sum_{i\in S}p_i$.

For a fixed $i\in S$ and a fixed bit tuple $b\in\{0,1\}^k$, independence of the players' inputs and the product-message assumption imply that the joint average message conditioned on $(x_{1,i},\ldots,x_{k,i})=b$ is the product state $\bigotimes_{j=1}^k\rho_{j,i}^{b_j}$. The POVM element that outputs $(i,b)$ identifies this tuple unambiguously. Applying \Cref{cor:avg-capacity} and then averaging over the uniformly distributed tuple $b$ gives
\[
    p_i \le 2^k a_k^{(i)}\prod_{j=1}^{k-1}a_j^{(i)}.
\]
Therefore, the required global success probability implies
\[
    1 - \varepsilon \le \sum_{i \in S} p_i \le 2^k \sum_{i \in S} \left( a_k^{(i)} \prod_{j=1}^{k-1} a_j^{(i)} \right).
\]

To bridge the product of local capacities to our pooled sum $c_i$, we apply the Arithmetic Mean-Geometric Mean (AM-GM) inequality to the first $k-1$ players. By isolating the product and raising both the arithmetic and geometric means to the power of $k-1$, we obtain:
\[
    \prod_{j=1}^{k-1} a_j^{(i)} \le \left( \frac{1}{k-1} \sum_{j=1}^{k-1} a_j^{(i)} \right)^{k-1} = \left( \frac{c_i}{k-1} \right)^{k-1} \le \left( \frac{2M}{(k-1)n} \right)^{k-1}.
\]

\noindent{\bf Wrapping Up the Lower Bound.}
Substituting this uniform bound back into the global success sum yields:
\[
    1 - \varepsilon \le 2^k \sum_{i \in S} a_k^{(i)} \left( \frac{2M}{(k-1)n} \right)^{k-1} \le 2^k m_k \left( \frac{2M}{(k-1)n} \right)^{k-1}.
\]

Finally, we evaluate this expression asymptotically. Substituting $M \le (k-1)m$ and $m_k \le m$ gives:
\[
    1 - \varepsilon \le 2^k m \left( \frac{2m}{n} \right)^{k-1}.
\]
Rearranging the terms to isolate $m^k$ yields
\[
    m^k
    \ge
    (1-\varepsilon)2^{-k}
    \left(\frac{1}{2}\right)^{k-1}
    n^{k-1}.
\]
Taking the $k$-th root gives the explicit bound
\[
    m
    \ge
    \frac{1}{2}(1-\varepsilon)^{1/k}
    \left(\frac{1}{2}\right)^{\frac{k-1}{k}}
    n^{1-\frac{1}{k}}.
\]

For every $k\ge 2$, we have $(1-\varepsilon)^{1/k}\ge (1-\varepsilon)^{1/2}$ and
\[
    \left(\frac{1}{2}\right)^{\frac{k-1}{k}}
    \ge
    \frac{1}{2}.
\]
Consequently,
\[
    m
    \ge
    \frac{\sqrt{1-\varepsilon}}{4}
    n^{1-\frac{1}{k}}
    =
    \Omega\left(n^{1-\frac{1}{k}}\right).
\]
\end{proof}

\section{Conclusion and Open Problems}

The results established in this work identify several settings in which public randomness is more powerful than quantum communication without shared randomness or prior entanglement in the multiparty SMP model.

We first analyze the total Boolean function $\operatorname{XEQ}_{k,n}$, obtaining a bounded-error separation in the same resource direction without imposing an input promise. This extends the phenomenon established for the multiparty promise function $GP_k$ by Gavinsky, Ito, and Wang~\cite{gavinsky2013} to a total Boolean problem, with a growing number of players. To our knowledge, this is the first separation for a total Boolean function in the SMP model in which public-coin classical communication outperforms quantum communication without shared randomness or prior entanglement by a polynomial factor.

We then turn to relations. By introducing and analyzing the $k$-player relation $\operatorname{IC}_{k,n}$, we obtain polynomial lower bounds on the maximum quantum message length that, together with the logarithmic public-coin upper bound, establish exponential separations in both the bounded-error and unambiguous regimes. We also construct a classical private-coin protocol that matches the unambiguous quantum lower bound, determining the communication complexity in this regime up to constant factors.

The relation $\operatorname{IC}_{k,n}$ provides a natural setting in which the coordination supplied by public randomness has a decisive advantage. In the unambiguous regime, arbitrary local quantum encodings and a joint measurement at the referee offer no asymptotic improvement over classical private-coin protocols. Shared randomness, however, allows the players to coordinate their choice of an index and thereby reduces the communication exponentially. This contrast isolates the value of coordination and shows that the additional capabilities of quantum messages cannot always compensate for its absence.

Despite these results, several fundamental questions remain open concerning the precise power of quantum communication without shared randomness or prior entanglement, the optimal bounds for the problems studied here, and the limitations of multipartite quantum state discrimination.

\noindent{\bf A Multiparty Direct Product Theorem for Bounded-Error Identification.}
To derive our bounded-error lower bound, we circumvented the lack of a multiparty direct product theorem by employing a parity-collapse reduction. This allowed us to group the first $k-1$ players into a single meta-register and apply existing two-register asymmetric direct product theorems. Extending such bounds to the identification of all $k$ labels directly could provide another approach to understanding the bounded-error complexity of multiparty coordination.

\begin{open}
    Is it possible to obtain a generalized direct product theorem that directly bounds the identification probabilities of a target sequence encoded across a $k$-fold tensor product of mixed states in the bounded-error regime?
\end{open}

\noindent{\bf Unifying the Lower Bounds Across Error Regimes.}
For fixed $k$, an intriguing mathematical gap remains between our two quantum lower bounds: we established a tight $\Omega(n^{1-1/k})$ lower bound for the unambiguous regime, but only an $\Omega\left(n^{\frac{k-1}{k+1}}\right)$ bound for the bounded-error regime. These exponents arise from two different specializations of the RAC inequality. In the unambiguous setting, the zero-error requirement sets $\eta_i=0$ and yields an $L_1$ linear constraint on local identification probabilities. In the bounded-error setting, applying RAC to the biases of Helstrom predictors yields an $L_2$ sum-of-squares constraint on local trace distances. Whether this difference reflects the complexity of the problem or a limitation of the proof remains open.

\begin{open}
    For fixed $k$, can the bounded-error quantum lower bound for $\operatorname{IC}_{k,n}$ be improved to match the $\Omega(n^{1-1/k})$ unambiguous bound, or does there exist an unentangled quantum protocol that exploits the permitted error margin to achieve $O\left(n^{\frac{k-1}{k+1}}\right)$ communication?
\end{open}

\bibliography{refs}

\appendix
\crefalias{section}{appendix}

\newpage

\section{Proof of \texorpdfstring{\Cref{thm:product-bound-gen}}{Theorem \ref*{thm:product-bound-gen}}}
\label{app:proof-of-dir}

We first establish an algebraic technical lemma demonstrating that the orthogonal complement of the sum of ``wrong'' tensor product spaces perfectly factorizes~\cite[Theorem~9]{gupta2024}.

\begin{lemma}[Generalized Subspace Factorization]\label{lem:subspace-gen}
Let $k \in \mathbb{N}$. For each $j \in [k]$, let $\mathcal{H}_j$ be a finite-dimensional Hilbert space, and let $\Sigma_j$ be a finite alphabet. Suppose $\mathcal{H}_j = \sum_{a \in \Sigma_j} S_j^{(a)}$ where $S_j^{(a)} \subseteq \mathcal{H}_j$ are closed subspaces. Here and throughout the statement, a sum of subspaces denotes their linear span and need not be direct.
For any target string $b \in \Sigma_1 \times \cdots \times \Sigma_k$, define the global ``wrong'' subspace $W_b$ as the sum of the product subspaces indexed by label strings $b'\neq b$:
\[
    W_b = \sum_{b' \neq b} \bigotimes_{j=1}^k S_j^{(b'_j)}.
\]
Then the orthogonal complement of $W_b$ in the full tensor space $\mathcal{H} = \bigotimes_{j=1}^k \mathcal{H}_j$ factorizes completely:
\[
    W_b^{\perp} = \bigotimes_{j=1}^k \left( \sum_{a \neq b_j} S_j^{(a)} \right)^{\perp}.
\]
\end{lemma}

\begin{proof}
For each register $j \in [k]$, define the local ``wrong'' subspace
\[
    T_{j,b_j}:=\sum_{a\neq b_j}S_j^{(a)}.
\]
For every $t\in[k]$, we define the local cylinders
\[
    C_t
    :=
    \mathcal{H}_1\otimes\cdots\otimes\mathcal{H}_{t-1}
    \otimes T_{t,b_t}\otimes
    \mathcal{H}_{t+1}\otimes\cdots\otimes\mathcal{H}_k.
\]

We first show that
\[
    W_b=\sum_{t=1}^k C_t.
\]
Indeed, every string $b'\neq b$ differs from $b$ in at least one coordinate $t$. For such a coordinate,
$S_t^{(b'_t)}\subseteq T_{t,b_t}$, and hence
\[
    \bigotimes_{j=1}^k S_j^{(b'_j)}
    \subseteq C_t.
\]
Summing over all $b'\neq b$ gives $W_b\subseteq\sum_{t=1}^k C_t$. Conversely, the assumptions
$\mathcal{H}_j=\sum_{a\in\Sigma_j}S_j^{(a)}$ and
$T_{t,b_t}=\sum_{a\neq b_t}S_t^{(a)}$, together with distributivity of tensor products over finite sums, imply
\[
    C_t
    =
    \sum_{\substack{a\in\Sigma_1\times\cdots\times\Sigma_k\\ a_t\neq b_t}}
    \bigotimes_{j=1}^k S_j^{(a_j)}
    \subseteq W_b.
\]
Therefore, the two subspaces are equal.

We next compute the orthogonal complement of each cylinder $C_t$. The orthogonal decomposition
$\mathcal{H}_t=T_{t,b_t}\oplus T_{t,b_t}^{\perp}$ induces the orthogonal decomposition
\[
    \mathcal{H}
    =
    C_t
    \oplus
    \left(
    \mathcal{H}_1\otimes\cdots\otimes\mathcal{H}_{t-1}
    \otimes T_{t,b_t}^{\perp}\otimes
    \mathcal{H}_{t+1}\otimes\cdots\otimes\mathcal{H}_k
    \right).
\]
Consequently, $C_t^\perp
    =
    \mathcal{H}_1\otimes\cdots\otimes\mathcal{H}_{t-1}
    \otimes T_{t,b_t}^{\perp}\otimes
    \mathcal{H}_{t+1}\otimes\cdots\otimes\mathcal{H}_k.$ Finally, the orthogonal complement of a finite sum of subspaces is the intersection of their orthogonal complements, and therefore
\[
    W_b^\perp
    =
    \left(\sum_{t=1}^k C_t\right)^\perp
    =
    \bigcap_{t=1}^k C_t^\perp.
\]
For each $t$, the orthogonal projector onto $C_t^\perp$ is
\[
    Q_t
    =
    I_{\mathcal{H}_1}\otimes\cdots\otimes I_{\mathcal{H}_{t-1}}
    \otimes\Pi_{T_{t,b_t}^{\perp}}\otimes
    I_{\mathcal{H}_{t+1}}\otimes\cdots\otimes I_{\mathcal{H}_k}.
\]
These projectors act nontrivially on distinct tensor factors and hence
commute pairwise. The product of a finite family of commuting orthogonal
projectors is the orthogonal projector onto the intersection of their
ranges; see \cite{Halmos1974FiniteDimensional} for a proof of this fact.
Therefore,
\[
    \Pi_{W_b^\perp}
    =
    \prod_{t=1}^k Q_t
    =
    \bigotimes_{j=1}^k\Pi_{T_{j,b_j}^{\perp}}.
\]
Taking ranges on both sides gives
\[
    W_b^\perp
    =
    \bigotimes_{j=1}^k T_{j,b_j}^{\perp}
    =
    \bigotimes_{j=1}^k
    \left(\sum_{a\neq b_j}S_j^{(a)}\right)^\perp,
\]
which concludes the proof.
\end{proof}

\begin{proof}[Proof of \Cref{thm:product-bound-gen}]
We now proceed with the proof of the theorem. Let $S_j^{(a)} = \operatorname{supp}(\alpha_j^{(a)})$. We can restrict each local space to the sum of its supports~\cite{rudolph2003}, so $\mathcal{H}_j = \sum_{a \in \Sigma_j} S_j^{(a)}$.

First, consider the local identification problem. To identify the state $\alpha_j^{(b_j)}$ locally without error, a measurement operator $E_{b_j}$ must satisfy $\operatorname{Tr}(E_{b_j} \alpha_j^{(a)}) = 0$ for all incorrect symbols $a \neq b_j$. Since both operators are positive semidefinite, this is equivalent to requiring their supports to be orthogonal. Thus, defining the local ``wrong'' subspace as $T_{j, b_j} = \sum_{a \neq b_j} S_j^{(a)}$, every feasible operator $E_{b_j}$ satisfies $\operatorname{supp}(E_{b_j})\subseteq T_{j,b_j}^{\perp}$. Moreover, because $E_{b_j}$ is a POVM element, $0\preceq E_{b_j}\preceq I$, and hence $E_{b_j}\preceq\Pi_{T_{j,b_j}^{\perp}}$. The projector $\Pi_{T_{j,b_j}^{\perp}}$ is itself feasible: set every other conclusive POVM element to zero and take $E_{\bot}=I-\Pi_{T_{j,b_j}^{\perp}}$. Therefore, the optimal local success probability is
\[
    p_j^{(b_j)} = \operatorname{Tr}\left( \alpha_j^{(b_j)} \Pi_{T_{j, b_j}^{\perp}} \right).
\]

Now, consider the global identification problem. To unambiguously identify the target sequence $b$, the global POVM element $E_{\text{global}}^{(b)}$ must never trigger if any incorrect product state $\rho_{b'}$ (where $b' \neq b$) was sent. This enforces the strict condition:
\[
    \forall b' \neq b, \quad \operatorname{Tr}\left( E_{\text{global}}^{(b)} \, \rho_{b'} \right) = 0.
\]
Because $\rho_{b'} = \bigotimes_{j=1}^k \alpha_j^{(b'_j)}$, the support of $\rho_{b'}$ is exactly the tensor product of the local supports: $\bigotimes_{j=1}^k S_j^{(b'_j)}$.
For the trace to be zero, the measurement operator $E_{\text{global}}^{(b)}$ must be orthogonal to the support of every single incorrect product state. By extension, it must be orthogonal to the entire subspace spanned by the union of all these incorrect supports.

Let $W_b$ denote this global ``wrong'' subspace, defined explicitly as the sum of the supports of all incorrect product states:
\[
    W_b = \sum_{b' \neq b} \operatorname{supp}(\rho_{b'}) = \sum_{b' \neq b} \bigotimes_{j=1}^k S_j^{(b'_j)}.
\]
Consequently, the support of $E_{\text{global}}^{(b)}$ must be contained in $W_b^{\perp}$. Since $0\preceq E_{\text{global}}^{(b)}\preceq I$, it follows that $E_{\text{global}}^{(b)}\preceq\Pi_{W_b^{\perp}}$. Therefore, the maximum global success probability is bounded by
\[
    P(b) \le \operatorname{Tr}\left( \rho_b \, \Pi_{W_b^{\perp}} \right).
\]

At this point, we apply~\Cref{lem:subspace-gen}, which establishes that the orthogonal complement of this global sum of product subspaces factorizes into the tensor product of the local orthogonal complements:
\[
    W_b^{\perp} = \bigotimes_{j=1}^k T_{j, b_j}^{\perp}.
\]
Substituting this factorization into the preceding trace bound and using multiplicativity of the trace over tensor products, we obtain:
\[
    P(b) \le \operatorname{Tr}\left( \left(\bigotimes_{j=1}^k \alpha_j^{(b_j)}\right) \bigotimes_{j=1}^k \Pi_{T_{j, b_j}^{\perp}} \right) = \prod_{j=1}^k \operatorname{Tr}\left( \alpha_j^{(b_j)} \Pi_{T_{j, b_j}^{\perp}} \right) = \prod_{j=1}^k p_j^{(b_j)}.
\]
This completes the proof.
\end{proof}

\end{document}